\documentclass[a4paper,final]{scrartcl}

\usepackage[utf8]{inputenc}
\usepackage[T1]{fontenc}

\usepackage{microtype}
\usepackage{nccmath}
\usepackage[nocompress]{cite}
\usepackage{tikz}
\usepackage{tabularx}
\usepackage{booktabs}
\usepackage{enumitem}

\usepackage{lmodern}
\usepackage{amssymb}
\usepackage{amsthm}
\usepackage{hyperref}

\renewcommand{\geq}{\geqslant}
\renewcommand{\leq}{\leqslant}
\renewcommand{\ge}{\geq}
\renewcommand{\le}{\leq}

\newcommand{\union}{\mathbin{\cup}}
\newcommand{\intersect}{\mathbin{\cap}}

\newcommand{\abs}[1] {\ensuremath\left|#1\right|}
\newcommand{\ceil}[1] {\ensuremath\left\lceil#1\right\rceil}
\newcommand{\set}[2] {\ensuremath{\left\{#1 \mid #2\right\}}}
\newcommand{\os}[1] {\ensuremath{\left\{#1\right\}}}

\newcommand{\bigO}{\mathcal{O}}

\newcommand{\FOm}{\ensuremath{\mathrm{FO}[{{\mkern 2mu\cdot\mkern 2mu}}]}}

\newcommand{\AC}{\ensuremath{\mathsf{AC}}}
\newcommand{\ACC}{\ensuremath{\mathsf{ACC}}}
\newcommand{\TC}{\ensuremath{\mathsf{TC}}}
\newcommand{\NC}{\ensuremath{\mathsf{NC}}}
\newcommand{\LOGSPACE}{\ensuremath{\mathsf{L}}}
\newcommand{\NL}{\ensuremath{\mathsf{NL}}}
\newcommand{\PTIME}{\ensuremath{\mathsf{P}}}
\newcommand{\NP}{\ensuremath{\mathsf{NP}}}

\newcommand{\vA}{\mathbf{A}}
\newcommand{\vB}{\mathbf{B}}
\newcommand{\vC}{\mathbf{C}}
\newcommand{\vD}{\mathbf{D}}
\newcommand{\vG}{\mathbf{G}}
\newcommand{\vGoid}{\mathbf{Goid}}
\newcommand{\vH}{\mathbf{H}}
\newcommand{\vI}{\mathbf{I}}
\newcommand{\vK}{\mathbf{K}}
\newcommand{\vM}{\mathbf{M}}
\newcommand{\vNil}{\mathbf{N}}
\newcommand{\vO}{\mathbf{O}}
\newcommand{\vP}{\mathbf{PGoid}}

\newcommand{\vS}{\mathbf{S}}
\newcommand{\vV}{\mathbf{V}}

\newcommand{\oD}{\mathbb{D}}
\newcommand{\oE}{\mathbb{E}}
\newcommand{\oL}{\mathbb{L}}

\newcommand{\vDV}{\mathbf{\oD{V}}}
\newcommand{\vEA}{\mathbf{\oE{A}}}
\newcommand{\vEV}{\mathbf{\oE{V}}}
\newcommand{\vLG}{\mathbf{\oL{G}}}
\newcommand{\vLI}{\mathbf{\oL{I}}}
\newcommand{\vLV}{\mathbf{\oL{V}}}

\newcommand{\gR}{\mathcal{R}}
\newcommand{\gL}{\mathcal{L}}
\newcommand{\gJ}{\mathcal{J}}
\newcommand{\gH}{\mathcal{H}}
\newcommand{\Rle}{\le_{\gR}}
\newcommand{\Lle}{\le_{\gL}}
\newcommand{\Jle}{\le_{\gJ}}
\newcommand{\Hle}{\le_{\gH}}

\newcommand{\Req}{\mathrel{\ensuremath{\gR}}}
\newcommand{\Leq}{\mathrel{\ensuremath{\gL}}}
\newcommand{\Jeq}{\mathrel{\ensuremath{\gJ}}}
\newcommand{\Heq}{\mathrel{\ensuremath{\gH}}}

\newcommand{\rRM}{\mathrel{\mathsf{RM}}}
\newcommand{\rLM}{\mathrel{\mathsf{LM}}}
\newcommand{\rGGM}{\mathrel{\mathsf{GGM}}}
\newcommand{\rAGGM}{\mathrel{\mathsf{AGGM}}}

\newcommand{\malcev}{\mathbin{\smash{\text{\small\ensuremath{%
  \tikz[baseline=-.75ex]{\node [shape=circle,draw,inner sep=0.4pt,scale=0.75] (m) {\ensuremath{\hspace*{0.1mm}m}};}
}}}}}
\newcommand{\tsp}{\ensuremath\mathbin{{*}\kern-.05ex{*}}}

\newcommand{\Parity}{\ensuremath{\textsc{Parity}}}

\newcommand{\ie}{i.e.,~}
\newcommand{\eg}{e.g.~}
\newcommand{\ms}{\hspace*{0.5pt}}

\setlist[itemize]{leftmargin=*,itemsep=-1ex}
\setlist[enumerate]{leftmargin=*,itemsep=-1ex}

\theoremstyle{plain}
\newtheorem{theorem}{Theorem}
\newtheorem{lemma}[theorem]{Lemma}
\newtheorem{corollary}[theorem]{Corollary}
\newtheorem{proposition}[theorem]{Proposition}

\title{Efficient Membership Testing for \protect\\ Pseudovarieties of Finite Semigroups}

\author{Lukas Fleischer}
\date{FMI, University of Stuttgart\thanks{Supported by the German Research Foundation (DFG) under grant DI 435/5--2.} \\
  Universitätsstraße 38, 70569 Stuttgart, Germany\\
\texttt{fleischer@fmi.uni-stuttgart.de}}

\begin{document}

\maketitle

\begin{abstract}
  \noindent
  {\sffamily\normalsize\bfseries{Abstract.}} \
  We consider the complexity of deciding membership of a given finite semigroup
  to a fixed pseudovariety.
  While it is known that there exist pseudovarieties with $\NP$-complete or
  even undecidable membership problems, for many well-known pseudovarieties the
  problem is known to be decidable in polynomial time.
  We show that for many of these pseudovarieties, the membership problem is
  actually in $\AC^0$.
  To this end, we show that these pseudovarieties can be characterized by
  first-order sentences with multiplication as the only predicate.
  We prove closure properties of the class of pseudovarieties with such
  first-order descriptions under various well-known operations; in particular,
  if $\vV$ can be described by a first-order sentence, then $\vDV$, $\vLV$,
  $\vK \malcev \vV$, $\vD \malcev \vV$, $\vNil \malcev \vV$, $\vLI \malcev
  \vV$, $\vLG \malcev \vV$ are first-order definable as well.
  Moreover, if $\vH$ is a first-order definable pseudovariety of finite groups,
  then $\overline \vH$ is first-order definable.
  Our formalism is also powerful enough to capture all pseudovarieties
  characterized by finite sets of $\omega$-identities.
  In view of lower bounds from circuit complexity, we obtain a new technique to
  prove that a pseudovariety $\vV$ cannot be defined by such a set: if
  membership in~$\vV$ is hard for $\Parity$, it cannot be defined in this logic
  and thus cannot be described by finitely many $\omega$-identities.
  We show that membership to $\vEA$ is $\LOGSPACE$-complete, thereby improving
  previous complexity results and obtaining a new proof that the pseudovariety
  cannot be described by finitely many $\omega$-identities at the same time.
\end{abstract}

\section{Introduction}

Over the past fifty years, algorithmic questions from language theory gave rise
to an intensive study of \emph{pseudovarieties} of finite semigroups.
These pseudovarieties are classes of finite semigroups closed under taking
(finite) direct products, subsemigroups and quotients. The concept was
introduced by Eilenberg~\cite{eil76} to formalize the correspondence between
natural subsets of the regular languages and finite semigroups.
Perhaps one of the most important algorithmic problems in the area of
pseudovarieties is the \emph{membership problem}, which asks for a given finite
semigroup $S$ whether or not it belongs to some fixed pseudovariety $\vV$.
It was shown that many natural operations on pseudovarieties do not preserve
decidability of the membership problem~\cite{abr92jsl,Rhodes99,Auinger10}.
Later, people started investigating the computational complexity of the
problem~\cite{Almeida1991,Volkov93,alm94,KharlampovichSapir94,Volkov97,AlmeidaBKK15,Vernitkii}.
While it has been well-known that for many commonly appearing pseudovarieties,
membership is decidable in polynomial time, it was conjectured that not all
pseudovarieties have this property.
Indeed, in 2006, Jackson and McKenzie constructed a pseudovariety for which
membership cannot be decided in polynomial time unless $\PTIME =
\NP$~\cite{JacksonMcKenzie2006}. Similar constructions followed; see
\eg\cite{Jackson2010}.

One of the main tools to show that the membership problem for a given
pseudovariety~$\vV$ can be decided in polynomial time is to show that $\vV$
can be defined by a finite set of so-called \emph{$\omega$-identities}. These
identities can be thought of as a system of equations with variables and terms
being built upon the binary operation in the semigroup and a special unary
operator, the so-called $\omega$-operator.
The na{\"\i}ve algorithm for testing membership to a pseudovariety defined by a
finite set of such identities is simply taking all possible assignments of
elements to the variables in the identities, computing the left-hand and
right-hand sides and testing whether all equalities hold.
Each left-hand and right-hand side can actually be computed within logarithmic
space. Thus, being definable by a finite set of \emph{$\omega$-identities} also
yields decidability in logarithmic space~\cite[Theorem
2.19]{StraubingW15arxiv}.

We continue this line of research by conducting a much more fine-grained
complexity analysis of pseudovarieties with polynomial-time decidable
membership.
We classify these pseudovarieties into natural complexity classes within
$\PTIME$ and we show that all for pseudovarieties defined by finitely many
$\omega$-identities, membership is decidable in $\AC^0$.
It is known that $\AC^0$ is not powerful enough to decide
$\Parity$~\cite{Hastad86,Yao85}.
Therefore, as a corollary, we obtain that a pseudovariety with membership hard
for any class containing $\Parity$ (such as $\ACC^0$, $\TC^0$, $\LOGSPACE$ or
$\NL$) cannot be defined by a finite set of $\omega$-identities.
We show that the membership problem for the pseudovariety $\vEA$, which was
proven to be not finitely based yet decidable in polynomial
time~\cite{Volkov93}, is $\LOGSPACE$-complete, thereby obtaining both a strict
improvement in complexity and a new and simple proof that this pseudovariety
cannot be defined by finitely many $\omega$-identities.
The $\LOGSPACE$-hardness proof is easy and relies on a technique reminiscent of
ideas appearing in~\cite{Mashevitzky83,Volkov93,Volkov96}.
The algorithm for proving containment in $\LOGSPACE$ is based on the notion of
the incidence graph of a Rees matrix semigroup as introduced by
Graham~\cite{Graham68} and on a famous result by Reingold~\cite{Reingold08}
which states that reachability in undirected graphs is decidable in
deterministic log-space.

Our main technique for establishing decidability of membership in $\AC^0$ is
interesting in its own right: we show that such pseudovarieties are definable
by first-order formulas with multiplication as the only predicate.
Our complexity bounds hold for all classes of finite semigroups definable using
this formalism which is, a priori, more powerful than finite sets of
$\omega$-identities.
We prove closure properties of the class of finite semigroups definable in this
logic fragment, such as closure under the $\oD(\cdot)$-operator, the
$\oL(\cdot)$-operator, the $\overline {\vphantom{\vH}(\cdot)}$-operator and
Mal'cev products with various pseudovarieties.
The logical formalism is a very natural way of defining classes of finite
semigroups, which makes the investigation of their expressiveness an
interesting subject even outside the scope of complexity questions.

Lastly, we investigate the complexity of deciding variety membership of
\emph{languages}, a problem which has been considered for various varieties and
for different representations of the languages
before~\cite{pw97,chohuynh91,GlasserSS16}.
This problem is shown to be $\NL$-hard when the language is given by morphisms
to finite semigroups, even for non-trivial language varieties for which the
corresponding pseudovariety membership problem is in~$\AC^0$.

In some of our results, we consider finite partial groupoids instead of finite
semigroups. There are two reasons for considering this generalization.
Firstly, the inputs are typically given as multiplication tables, and a priori,
we cannot guarantee that the operation given by such a multiplication table is
associative or that all entries are valid encodings of elements. Secondly, when
proving closure under certain operations, we need to consider subsets of
semigroups which are not subsemigroups. For example, for closure under the
$\oD(\cdot)$-operator, we need to restrict formulas to the regular
$\gJ$-classes of a semigroup which, in general, do not form subsemigroups (but
can be considered as partial subsemigroups).

\section{Preliminaries}

\paragraph{Algebra.}

A \emph{partial groupoid} is a non-empty set endowed with a partial binary
operation.
If the operation is associative, the structure is a \emph{partial semigroup}.
A partial groupoid with a total operation is called \emph{groupoid} (not to be
confused with the notion of groupoids arising in category theory).
If, additionally, the operation is associative, the structure is called
\emph{semigroup}.
A \emph{monoid} is a semigroup with an identity element, and if every element
has a~(unique) inverse, the monoid is a \emph{group}.
The class of all finite partial groupoids is denoted by $\vP$, the class of all
finite groupoids is denoted by $\vGoid$, the class of all finite semigroups is
denoted by $\vS$, the class of all finite monoids is denoted by $\vM$ and the
class of all finite groups is denoted by $\vG$.
All algebraic structures considered in this work are either free or finite.
The operation of a partial groupoid is often also referred to as
\emph{multiplication}.

If $G$ is a partial groupoid with multiplication $\cdot \colon G \times G \to
G$ and $X$ is a subset of $G$, then $X$ forms a partial groupoid with the
multiplication $\circ \colon X \times X \to X$ given by $x \circ y = x \cdot y$
if $x \cdot y \in X$ and $x \circ y$ undefined otherwise.
The partial subgroupoid \emph{generated by a subset $X$} of a partial groupoid
$G$, denoted by $\langle X \rangle$, is the smallest subgroupoid of $G$
containing all elements of $X$ and closed under multiplication.
An element $g$ of a partial groupoid $G$ is \emph{idempotent} if $g \cdot g = g$.
The set of all idempotent elements of $G$ is denoted by $E(G)$.
We define a unary operator $\omega_G \colon G \to G$ as follows: for each
element $g \in G$, let $g^{(1)} = g$ and $g^{(i+1)} = g^{(i)} \cdot g$. Then,
if $\set{g^{(i)}}{i \ge 1}$ contains a unique idempotent element $h$, we let
$\omega_G(g) = h$. Otherwise, $\omega_G(g)$ is undefined. We often use the
notation $g^{\omega_G}$ instead of $\omega_G(g)$.

In this work, we are mainly interested in finite semigroups.
It is well-known that in a finite semigroup $S$, the element $s^{\omega_S}$ is
defined for each $s \in S$, see \eg\cite[Proposition 1.6]{pin86}.
Green's relations are a useful tool to study structural properties of finite
semigroups.
We denote by $S^1$ the monoid that is obtained by adding a new neutral element
$1$ to a finite semigroup $S$.
For $s, t \in S$ let
\begin{ceqn}
  \begin{align*}
    s & \Rle t \text{~if there exists~} q \in S^1 \text{~such that~} s = tq, &
    s & \Req t \text{~if~} s \Rle t \text{~and~} t \Rle s, \\
    s & \Lle t \text{~if there exists~} p \in S^1 \text{~such that~} s = pt, &
    s & \Leq t \text{~if~} s \Lle t \text{~and~} t \Lle s, \\
    s & \Jle t \text{~if there exist~} p, q \in S^1 \text{~such that~} s = ptq, &
    s & \Jeq t \text{~if~} s \Jle t \text{~and~} t \Jle s, \\
    s & \Hle t \text{~if~} s \Rle t \text{~and~} s \Lle t, &
    s & \Heq t \text{~if~} s \Hle t \text{~and~} t \Hle s.
  \end{align*}
\end{ceqn}
Note that Green's relations can also be defined for partial groupoids but many
important properties require associativity.
For example, in a semigroup, each of the relations defined above is transitive.
The relation $\gR$ (resp.~$\gL$, $\gJ$, $\gH$) then becomes an equivalence
relation and its equivalence classes are called \emph{$\gR$-classes}
(resp.~\emph{$\gL$-classes}, \emph{$\gJ$-classes}, \emph{$\gH$-classes}).
A $\gJ$-class is \emph{regular} if it contains an idempotent element.

A \emph{semigroup congruence} on a semigroup $S$ is an equivalence relation
$\sim$ on $S$ which is compatible with multiplication, \ie{}$x \sim y$ implies
$pxq \sim pyq$ for all $x, y \in S$ and $p, q \in S^1$.
Given a semigroup congruence $\sim$ on $S$, one can define the \emph{quotient}
$S / {\sim}$ which is the set of equivalence classes of $S$ modulo $\sim$
equipped with the canonical multiplication induced by the multiplication in
$S$.
A \emph{morphism} from a semigroup $S$ to a semigroup $T$ is a mapping $h
\colon S \to T$ such that $h(x)h(y) = h(xy)$ for all $x, y \in S$.
The \emph{direct product} of two semigroups $S$ and $T$ is the Cartesian
product $S \times T$ with componentwise multiplication.

A \emph{relation (of arity $k$) on partial groupoids} is a family of relations
${(R_G)}_{G \in \vP}$ such that $R_G \subseteq G^k$ for each partial groupoid
$G$. For a given relation on partial groupoids, usually denoted by $R$, we
write $R_G$ to denote the relation corresponding to the partial groupoid $G$.
A~\emph{semigroup congruence on partial groupoids} is a binary relation on
partial groupoids $\sim$ such that for each finite semigroup $S$, the relation
$\sim_S$ is a semigroup congruence.

\paragraph{Pseudovarieties.}

A \emph{pseudovariety of finite semigroups} is a non-empty class of finite
semigroups closed under taking finite direct products, subsemigroups and
quotients.
All pseudovarieties considered in this work are pseudovarieties of finite
semigroups. Thus, we usually refer to them simply as \emph{pseudovarieties}.
The following pseudovarieties will be used later:
\begin{itemize}
  \item $\vA$, the pseudovariety of finite aperiodic semigroups, \ie{}$x \Heq y$ implies $x = y$,
  \item $\vB$, the pseudovariety of all finite bands (all elements are idempotent),
  \item $\vD$, the pseudovariety of all definite semigroups, \ie{}$ex = e$ for all $e \in E(S)$ and $x \in S$,
  \item $\vG$, the pseudovariety of all finite groups,
  \item $\vI$, the pseudovariety containing only the trivial semigroup $\os{1}$,
  \item $\vK$, the pseudovariety of all co-definite semigroups, \ie{}$xe = e$ for all $e \in E(S)$, $x \in S$,
  \item $\vNil$, the pseudovariety of all nilpotent semigroups (all idempotents are zero elements),
  \item $\vO$, the pseudovariety of finite orthodox semigroups, \ie{}$e, f \in E(S)$ implies $ef \in E(S)$,
  \item $\vS$, the pseudovariety of all finite semigroups.
\end{itemize}

Proofs that the classes listed above are indeed pseudovarieties can be found in
many standard textbooks, see \eg\cite{pin86,rs09qtheory,alm94}.

For two pseudovarieties $\vV_1$ and $\vV_2$, the \emph{join} of $\vV_1$ and
$\vV_2$, usually denoted by $\vV_1 \lor \vV_2$, is the smallest pseudovariety
containing both $\vV_1$ and $\vV_2$. One can also define the semidirect product
$\vV_1 * \vV_2$, the two-sided semidirect product $\vV_1 \tsp \vV_2$ and the
Mal'cev product $\vV_1 \malcev \vV_2$.
The formal definitions of these operations are not required in this work and
we refer the interested reader to~\cite{alm94,rs09qtheory} for details.

We will also consider some commonly occurring unary operators on
pseudovarieties.
For a pseudovariety $\vV$ and a finite semigroup $S$, we have
\begin{align*}
  S & \in \vDV\quad\text{if all regular $\gJ$-classes of $S$ form subsemigroups from $\vV$}, \\
  S & \in \vEV\quad\text{if $\langle E(S) \rangle \in \vV$}, \\
  S & \in \vLV\quad\text{if for all $e \in E(S)$, the monoid $eSe$ belongs to $\vV$}.
\end{align*}
For a pseudovariety of finite groups $\vH$, we let $\overline \vH$ be the class
of finite semigroups whose regular $\gH$-classes are groups from $\vH$.
For any pseudovariety of finite semigroups $\vV$ and for any pseudovariety of
finite groups $\vH$, the classes $\vDV$, $\vEV$, $\vLV$ and $\overline \vH$ are
themselves pseudovarieties; see \eg\cite{rs09qtheory,alm94}.

For a set of variables $X$, the set of \emph{$\omega$-terms} over $X$ is
defined inductively as follows: every variable from $X$ is an $\omega$-term and
if $U$ and $V$ are $\omega$-terms, then so are $UV$ and $U^\omega$.
An \emph{$\omega$-identity} is an equation of the form $U = V$ where $U$ and
$V$ are $\omega$-terms. If neither $U$ nor $V$ contain a subterm of the form
$U^\omega$, the identity is called an \emph{equation}.
Every mapping $h \colon X \to S$ to a finite semigroup $S$ extends uniquely to
$\omega$-terms by $h(UV) = h(U) h(V)$ and $h(U^\omega) = {(h(V))}^{\omega_S}$.
Such a mapping \emph{satisfies} an $\omega$-identity $U = V$ if $h(U) = h(V)$.
A finite semigroup $S$ satisfies an $\omega$-identity if the identity is
satisfied by every mapping $h \colon X \to S$. A set of $\omega$-identities is
satisfied if each of the identities in the set is satisfied.
The class of finite semigroups \emph{defined} by a set of $\omega$-identities
is the class of all finite semigroups satisfying the given set.
It is well known that every class of finite semigroups defined by a (not
necessarily finite) set of $\omega$-identities is a pseudovariety~\cite{rei82}.

\paragraph{Complexity.}

We assume familiarity with standard definitions from circuit complexity.
The circuit complexity class $\AC^0$ is the class of languages decidable by
unbounded fan-in Boolean circuit families of depth $\bigO(1)$ and polynomial
size.
We allow NOT gates but do not count them when measuring the depth or the size
of a circuit.
The class of languages decidable by a deterministic (resp.~non-deterministic)
Turing machine using a logarithmic amount of work space is denoted by
$\LOGSPACE$ (resp.~$\NL$). We will also briefly refer to the complexity classes
$\ACC^0$ and $\TC^0$.

It is known that the \Parity~function cannot be computed by $\AC^0$ circuits.
This follows directly from H{\aa}stad's and Yao's famous lower bound
results~\cite{Hastad86,Yao85}, which state that the number of Boolean gates
required for a depth-$d$ circuit to compute $\Parity$ is exponential in
$n^{1/(d-1)}$.
By identifying formal languages with their indicator functions, $\Parity$ is
often also understood as the set of strings over $\os{0, 1}^*$ with an even
number of $1$s.

\paragraph{Model of Computation.}

We consider decision problems whose inputs are finite partial groupoids,
encoded as multiplication tables.
More precisely, a partial groupoid $G$ of cardinality $N$ is encoded by $N^2
\ceil{\log(N+1)}$ bits: the elements are represented by the integers $\os{1,
\dots, N}$ and the multiplication table is given by sequence of elements
occurring in the table in row-major order.
Some values in the input may refer to non-existent elements (in particular, the
values $0$ and $N+1, \dots, 2^{\ceil{\log(N+1)}} - 1$ are not valid encodings
of any element). Such entries are interpreted as the product being undefined.

\paragraph{First-order Logic.}

We consider $\FOm$-formulas over finite partial groupoids. Variables correspond
to elements of a finite partial groupoid and the only allowed predicate is the
binary operation ``$xy = z$''.

Formally, the set of \emph{$\FOm$-formulas} is then defined inductively: each
atomic formula of the form $xy = z$ with variables $x, y, z$ is a
$\FOm$-formula and if $\varphi$ and $\psi$ are $\FOm$-formulas and $x$ is a
variable, then $\varphi \lor \psi$, $\lnot \varphi$ and $\exists x \colon
\varphi$ are $\FOm$-formulas as well.
We will use common abbreviations such as $\forall x \colon \varphi$ instead of
$\lnot \exists x \colon \lnot \varphi$, $\varphi \to \psi$ instead of
$\lnot \varphi \lor \psi$ or $\varphi \leftrightarrow \psi$ instead of
$(\varphi \land \psi) \lor (\lnot \varphi \land \lnot \psi)$.
A variable is said to occur \emph{freely} in a $\FOm$-formula if it does not
appear in the scope of a quantifier.
Sometimes, we write $\varphi(x_1, \dots, x_k)$ to indicate that at most the
variables $x_1, \dots, x_k$ occur freely in $\varphi$. A \emph{sentence} is a
formula without any free variables.

The \emph{truth value} of a $\FOm$-formula $\varphi(x_1, \dots, x_k)$ in a
partial groupoid $G$ for \emph{assignments} of elements $g_1, \dots, g_k \in G$
to the variables $x_1, \dots, x_k$ is defined as usual, and we write $(G, g_1,
\dots, g_k) \models \varphi(x_1, \dots, x_k)$ if the assignment specified by
$G$ and by assigning each $g_i$ to the corresponding $x_i$ satisfies~$\varphi$.
Two $\FOm$-formulas are \emph{equivalent} if they are satisfied by exactly
the same structures.
The relation symbol $\cdot$ is interpreted as the binary operation in the
corresponding finite partial groupoid $G$: it consists of all tuples $(x, y,
z)$ such that $xy = z$ in $G$. When using the relation in formulas, we usually
write $xy = z$ instead of ${{}\cdot{}}(x, y, z)$.
A class of finite partial groupoids $\vC$ is \emph{defined} by a
$\FOm$-sentence $\varphi$ if $G \in \vC$ if and only if $G \models \varphi$.
A relation on partial groupoids $R$ of arity $k$ is \emph{defined} by a
$\FOm$-formula $\psi(x_1, \dots, x_k)$ if $(g_1, \dots, g_k) \in R_G$ if and
only if $(G, g_1, \dots, g_k) \models \psi(x_1, \dots, x_k)$.

\section{First-Order Definable Classes of Finite Partial Groupoids}
\label{sec:fo}

The objective of this section is to provide sufficient conditions for
$\FOm$-definability of pseudovarieties, alongside with some examples of classes
definable in these logic fragments.
In Table~\ref{tab:intro}, we give $\FOm$-formulas for some of the most
well-known families of algebraic structures. These formulas will be reused
throughout the paper and referenced by the identifiers given in the table.
It is easy to see that the class of $\FOm$-definable classes of finite partial
groupoids is closed under Boolean operations.

\begin{table}
  \begin{tabularx}{\textwidth}{p{15mm}lX}
    \toprule
    Class & \multicolumn{2}{l}{$\FOm$-formula} \\
    \midrule
    $\vP$ & $\varphi_\vP :=$ & $\exists x \colon x = x$ \\
    $\vGoid$ & $\varphi_\vGoid :=$ & $\forall x \forall y \exists z \colon xy = z$ \\
    $\vS$ & $\varphi_\vS :=$ & $\varphi_\vGoid \land \forall x \forall y \forall z \exists u \exists v \exists w \colon xy = u \land uz = v \land xw = v \land w = yz$ \\
    $\vI$ & $\varphi_\vI :=$ & $\varphi_S \land \forall x \forall y \colon x = y$ \\
    $\vM$ & $\varphi_\vM := $ & $\varphi_\vS \land \exists x \forall y \colon xy = y \land yx = y$ \\
    $\vG$ & $\varphi_\vG := $ & $\varphi_\vM \land \exists x \forall y \exists z \colon xy = y \land yx = y \land yz = x \land zy = x$ \\
    $\vB$ & $\varphi_\vB := $ & $\varphi_\vS \land \forall x \colon xx = x$ \\
    $\vO$ & $\varphi_\vO := $ & $\varphi_\vS \land \forall x \forall y \colon (xx = x \land yy = y) \to \exists z \colon xy = z \land zz = z$ \\
    \bottomrule
  \end{tabularx}
  \caption{Common classes of finite partial groupoids and corresponding $\FOm$-formulas}
  \label{tab:intro}
\end{table}

\begin{lemma}
  Let $\vC_1, \vC_2$ be $\FOm$-definable classes of finite partial groupoids.
  Then $\vC_1 \union \vC_2$, $\vC_1 \intersect \vC_2$ and $\vC_1 \setminus
  \vC_2$ are $\FOm$-definable.
  \label{lem:bool}
\end{lemma}
\begin{proof}
  If $\vC_1$ is defined by $\varphi_1$ and $\vC_2$ is defined by $\varphi_2$,
  then $\vC_1 \union \vC_2$ is defined by $\varphi_1 \lor \varphi_2$, $\vC_1
  \intersect \vC_2$ is defined by $\varphi_1 \land \varphi_2$, and $\vC_1
  \setminus \vC_2$ is defined by $\varphi_1 \land \lnot \varphi_2$.
\end{proof}

One can also show that Green's relations may be used in $\FOm$-formulas without
changing the expressive power of this logical fragment.

\begin{lemma}
  Green's relations $\Rle$, $\Lle$, $\Jle$, $\Hle$, $\gR$, $\gL$, $\gJ$ and
  $\gH$ are $\FOm$-definable.
  \label{lem:green}
\end{lemma}
\begin{proof}
  We have $x \Rle y$ if and only if $x = y \lor \exists z \colon yz = x$ and
  $x \Req y$ if and only if $x \Rle y$ and $y \Rle x$.
  A similar construction can be used for the other relations.
\end{proof}

We will make use of this observation and henceforth, Green's relations will be
used as abbreviations in some $\FOm$-formulas without further explanation.
The next lemma shows how Green's relations can be used to express the
$\omega$-operator.

\begin{lemma}
  Let $S$ be a finite semigroup and let $s \in S$. Then $s^{\omega_S}$ is the
  unique $\Hle$-maximal idempotent element $t \in S$ such that $st \Req t$ and
  $ts \Leq t$.
  \label{lem:omega-green}
\end{lemma}
\begin{proof}
  Since $\Req$ is stable on the left, every element $t \in S$ with $st \Req t$
  satisfies $s^{i+1} t \Req s^i t$ for all $i \ge 0$. By transitivity, we
  obtain $s^i t \Req t$ for all $i \ge 0$. In particular, $s^{\omega_S} t \Req
  t$, and thus, $t \Rle s^{\omega_S}$. Symmetrically, $ts \Leq t$ yields $t
  \Lle s^{\omega_S}$. Therefore, we have $t \Hle s^{\omega_S}$ for every
  element $t \in S$ with $st \Req t$ and $ts \Leq t$.

  Now, since $s^{\omega_S + 1} \Heq s^{\omega_S}$, every $\Hle$-maximal element
  $t \in S$ with $st \Req t$ and $ts \Leq t$ satisfies $t \Heq s^{\omega_S}$.
  It is well-known (and easy to verify) that every $\gH$-class contains at most
  one idempotent element. Therefore, if $t$ is idempotent, we have $t =
  s^{\omega_S}$, as desired.
\end{proof}

It is now easy to see that the class of $\FOm$-definable classes generalizes
the class of pseudovarieties defined by a finite set of $\omega$-identities.
This idea is captured in the following proposition.

\begin{proposition}
  Every pseudovariety of finite semigroups defined by a finite set of
  $\omega$-identities is $\FOm$-definable.
  \label{prop:omega-identities}
\end{proposition}

The statement follows immediately from the fact that the class of
$\FOm$-definable pseudovarieties is closed under intersection (thus, we can
confine ourselves to the case of a single $\omega$-identity) and from the
following lemma.

\begin{lemma}
  Let $U = V$ be an $\omega$-identity with variables $X = \os{x_1, \dots,
  x_k}$.  Then there exists an $\FOm$-formula $\varphi(x_1, \dots, x_k)$ such
  that for all finite semigroups $S$ and for all $s_1, \dots, s_k \in S$, we
  have $(S, s_1, \dots, s_k) \models \varphi(x_1, \dots, x_k)$ if and only if
  $U = V$ is satisfied by the mapping $h \colon X \to S$ with $h(x_i) =
  s_i$.\label{lem:eq}
\end{lemma}
\begin{proof}
  Without loss of generality, we may assume that $V = y$ for some variable $y$.
  The general case then follows by introducing a new variable $x_{k+1} \not\in
  X$, and combining formulas $\psi(x_1, \dots, x_{k+1})$ for the equation $U =
  x_{k+1}$ as well as $\chi(x_1, \dots, x_{k+1})$ for $V = x_{k+1}$ to obtain
  the formula $\varphi = \exists x_{k+1} \colon \psi(x_1, \dots, x_{k+1}) \land
  \chi(x_1, \dots, x_{k+1})$ for $U = V$.

  If $U$ also consists of a single variable $x$, we define $\varphi$ as $x =
  y$.
  If $U = W_1 W_2$ for $\omega$-terms $W_1$ and $W_2$, we apply induction on
  the length of $U$ to obtain formulas $\psi_1(x_1, \dots, x_k, x_{k+1})$ for
  the $\omega$-identity $W_1 = x_{k+1}$ and $\psi_2(x_1, \dots, x_k, x_{k+2})$
  for $W_2 = x_{k+2}$. Again, $x_{k+1}$ and $x_{k+2}$ are new variables not in
  $X$.
  We then let
  \begin{equation*}
    \varphi \, = \, \exists x_{k+1} \exists x_{k+2} \colon \psi_1(x_1, \dots, x_k, x_{k+1}) \land \psi_2(x_1, \dots, x_k, x_{k+2}) \land x_{k+1} x_{k+2} = y.
  \end{equation*}
  If $U$ has the form $W^\omega$ for an $\omega$-term $W$, we apply induction
  on the length of $U$ to obtain a formula $\psi(x_1, \dots, x_k, x_{k+1})$ for
  the $\omega$-identity $W = x_{k+1}$ where $x_{k+1}$ is a new variable not in
  $X$.
  We then define $\varphi$ as $\exists x_{k+1} \colon \psi(x_1, \dots, x_k,
  x_{k+1}) \land \xi(x_{k+1}, y)$ where $\xi$ is the formula
  \begin{align*}
    yy = y \land (x_{k+1} y \Req y) \land (yx_{k+1} \Leq y) \land \forall z \colon & y = z \lor zz \ne z \lor \lnot y \Hle z \lor {}\\
                                                                                   & \lnot(x_{k+1}z \Req z) \lor \lnot(zx_{k+1} \Leq z).
  \end{align*}
  The correctness of this construction follows from
  Lemma~\ref{lem:omega-green}.
\end{proof}

For an $\omega$-identity $U = V$ with variables $\os{x_1, \dots, x_k}$,
applying the lemma yields a formula $\varphi(x_1, \dots, x_k)$. The sentence
$\forall x_1 \cdots \forall x_k \colon \varphi(x_1, \dots, x_k)$ then defines
the pseudovariety defined by $U = V$.
The statement of the lemma is actually a bit stronger: it allows us to use
$\omega$-identities as abbreviations within $\FOm$-formulas.
In the following, we will often make use of this technique to improve the
readability of formulas.

Next, we would like to prove some more sophisticated closure properties of
$\FOm$-definable classes of finite semigroups.
Unfortunately, it seems impossible to obtain very general results:
in~\cite{Rhodes99} it was shown that there exist $\FOm$-definable
pseudovarieties of finite semigroups $\vV_1$, $\vV_2$ such that none of $\vV_1
\lor \vV_2$, $\vV_1 \malcev \vV_2$, $\vV_1 * \vV_2$ and $\vV_1 \tsp \vV_2$ are
decidable.
Since all $\FOm$-definable classes of finite semigroups are decidable
(actually, in $\AC^0$, as we shall see in the next section), this implies that
the class of $\FOm$-definable pseudovarieties is not closed under these
operations.
However, we will show that at least for many common operations already known to
preserve decidability, $\FOm$-decidability is preserved as well.

Remember that every non-empty subset $X$ of a finite semigroup $S$ forms a
partial semigroup with the multiplication induced by $S$.

\begin{proposition}
  Let $\vC$ be a $\FOm$-definable class of finite semigroups and let $\sim$ be
  a $\FOm$-definable binary relation on partial groupoids.
  Then the class of all finite semigroups $S$ with the property that all
  (non-empty) partial semigroups $\set{t \in S}{s \sim_S t}$ belong to~$\vC$ is
  $\FOm$-definable.
  \label{prop:filter}
\end{proposition}

\begin{proof}
  Suppose that $\varphi$ is a $\FOm$-sentence with $G \models \varphi$ if and
  only if $G \in \vC$ and suppose that $\psi(x, y)$ is a $\FOm$-formula such
  that for all finite partial groupoids $G$ and for all $g, h \in G$, we have
  $g \sim_G h$ if and only if $(G, g, h) \models \psi(x, y)$.
  We construct a $\FOm$-sentence $\varphi''$ such that $S \models \varphi''$
  whenever $S$ is a finite semigroup and all partial semigroups $T_s := \set{t
  \in S}{s \sim_S t}$ belong to~$\vC$, that is, $T_s \models \varphi$.

  Let $x$ be a variable not appearing in $\varphi$.
  We successively replace all subformulas of the form $\exists y \colon \chi$
  in $\varphi$ by $\exists y \colon \psi(x, y) \land \chi$ to obtain a formula
  $\varphi'$. Note that $x$ occurs freely in $\varphi'$.
  We then let $\varphi'' = \varphi_\vS \land \forall x \colon \varphi'(x)$.

  In order to prove correctness of the construction, we prove a slightly
  stronger statement: suppose that the original formula $\varphi$ contained
  additional free variables $z_1, \dots, z_k$. We show that then, for all
  finite semigroups $S$, for all $s \in S$ and for all $u_1, \ldots, u_k \in
  T_s$, we have $(S, s, u_1, \ldots, u_k) \models \varphi'(x, z_1, \dots, z_k)$
  if and only if $(T_s, u_1, \ldots, u_k) \models \varphi(z_1, \dots, z_k)$.

  If $\varphi$ is an atomic formula, the claim clearly holds.
  If $\varphi$ has the form $\chi_1 \lor \chi_2$ or $\lnot \chi$, the statement
  holds by induction.
  Suppose now that $\varphi(z_1, \dots, z_k) = \exists y \colon \chi(y, z_1,
  \dots, z_k)$ for some $\FOm$-formula $\chi$.
  Then,
  \begin{equation*}
    \varphi'(x, z_1, \dots, z_k) \,=\, \exists y \colon \psi(x, y) \land \chi'(x, y, z_1, \dots, z_k)
  \end{equation*}
  where, by induction, we have that $(S, s, t, u_1, \ldots, u_k) \models
  \chi'(x, y, z_1, \dots, z_k)$ if and only if $(T_s, t, u_1, \ldots, u_k)
  \models \chi(y, z_1, \dots, z_k)$.
  Using this and the definition of $\varphi'$, we know that $(S, s, u_1,
  \ldots, u_k) \models \varphi'(x, z_1, \dots, z_k)$ if and only if there
  exists some $t \in S$ with $s \sim_S t$ such that $(T_s, t, u_1, \ldots, u_k)
  \models \chi(y, z_1, \dots, z_k)$. By the definition of $\varphi$, the latter
  property is, in turn, equivalent to $(T_s, u_1, \ldots, u_k) \models
  \varphi(z_1, \dots, z_k)$, as desired.
\end{proof}

\begin{corollary}
  Let $\vV$ be a $\FOm$-definable pseudovariety of finite semigroups and let
  $\vH$ be a $\FOm$-definable pseudovariety of finite groups. Then $\vDV$,
  $\vLV$ and $\overline \vH$ are $\FOm$-definable.\label{crl:DLbar}
\end{corollary}
\begin{proof}
  To see that $\vDV$ is $\FOm$-definable, we define a relation on finite
  partial groupoids $\sim$ by $x \sim_G y$ if and only if $xx = x \land x \Jeq
  y$. By Lemma~\ref{lem:green}, this relation on partial groupoids if
  $\FOm$-definable.
  Moreover, given a finite semigroup $S$, the non-empty sets $\set{t \in S}{s
  \sim_S t}$ correspond to the regular $\gJ$-classes of $S$. Thus, $\vDV$ is
  $\FOm$-definable by Proposition~\ref{prop:filter}.
  The same arguments can be used for $\vLV$ (with $x \sim_G y$ if $xx = x \land
  \exists z \colon xzx = y$) and for $\overline \vH$ (with $x \sim_G y$ if $xx
  = x \land x \Heq y$).
\end{proof}

A second series of closure properties follows from the following result.

\begin{proposition}
  Let $\vC$ be a $\FOm$-definable class of finite semigroups and let $\sim$
  be a $\FOm$-definable semigroup congruence on partial groupoids.
  Then the class of all finite semigroups $S$ with $S/{\sim_S} \in \vC$ is
  $\FOm$-definable.
  \label{prop:congruence}
\end{proposition}

\begin{proof}
  Suppose that $\varphi$ is a $\FOm$-sentence with $S \models \varphi$ if and
  only if $S \in \vC$ and suppose that $\psi(x, y)$ is a $\FOm$-formula such
  that for all partial groupoids $G$ and for all $g, h \in G$, we have $g
  \sim_G h$ if and only if $(G, g, h) \models \psi(x, y)$.
  We construct a $\FOm$-sentence $\varphi''$ such that $S \models \varphi''$
  whenever $S$ is a finite semigroup and $S/{\sim_S} \in \vC$, that is,
  $S/{\sim_S} \models \varphi$.

  A formula $\varphi'$ is obtained by successively replacing all atomic
  formulas of the form $xy = z$ in $\varphi$ with a new sub-formula $\exists w
  \colon xy = w \land \psi(w, z)$ where $w$ is a fresh variable not appearing
  in $\varphi$.
  We then let $\varphi'' = \varphi_\vS \land \varphi'$.

  As in the proof of Proposition~\ref{prop:filter}, we prove a slightly
  stronger statement with additional free variables: we show that for every
  finite semigroup $S$ and for all $s_1, \dots, s_k \in S$, one has $(S, s_1,
  \dots, s_k) \models \varphi'(z_1, \dots, z_k)$ if and only if $(S /
  {\sim_S}, [s_1], \dots, [s_k]) \models \varphi(z_1, \dots, z_k)$. Here, for
  an element $s \in S$, we use $[s]$ to denote the equivalence class of $s$
  modulo $\sim_S$.

  If $\varphi$ is an atomic formula $xy = z$, then the quotient $S / {\sim_S}$
  satisfies $\varphi$ if and only if classes $[s]$, $[t]$ and $[u]$ are
  assigned to $x, y, z$ such that $[s] [t] = [u]$, \ie{}$st \sim_S u$. This is
  equivalent to the existence of some element $r \in S$ such that $st = r$ and
  $r \sim_S u$, which is exactly what $\varphi' = \exists w \colon xy = w \land
  \psi(w, z)$ expresses.
  Induction yields the desired statement if $\varphi$ has the form $\chi_1 \lor
  \chi_2$, $\lnot \chi$ or $\exists y \colon \chi$.
\end{proof}

\begin{corollary}
  Let $\vV$ be a $\FOm$-definable pseudovariety of finite semigroups.
  Then $\vK \malcev \vV$, $\vD \malcev \vV$, $\vNil \malcev \vV$, $\vLI \malcev
  \vV$, $\vLG \malcev \vV$ and are $\FOm$-definable.
  \label{crl:malcev}
\end{corollary}
\begin{proof}
  It is known~\cite[Theorem 4.6.50]{rs09qtheory} that
  $S \in \vK \malcev \vV$ if and only if $S / {\rRM} \in \vV$,
  $S \in \vD \malcev \vV$ if and only if $S / {\rLM} \in \vV$,
  $S \in \vNil \malcev \vV$ if and only if $S / (\rRM \intersect \rLM) \in \vV$,
  $S \in \vLI \malcev \vV$ if and only if $S / {\rGGM} \in \vV$, and
  $S \in \vLG \malcev \vV$ if and only if $S / {\rAGGM} \in \vV$
  where $\rRM$, $\rLM$, $\rGGM$ and $\rAGGM$ are the semigroup congruences on
  partial groupoids defined by
  \begin{align*}
    s \rRM t & \quad\text{if}\quad \forall x \colon \rho(x) \to \big(\lnot (xs \Jeq x \lor xt \Jeq x) \lor xs = xt\big), \\
    s \rLM t & \quad\text{if}\quad \forall x \colon \rho(x) \to \big(\lnot (sx \Jeq x \lor tx \Jeq x) \lor sx = tx\big), \\
    s \rGGM t & \quad\text{if}\quad \forall x \forall y \colon (\rho(x) \land x \Jeq y) \to \big(\lnot (xsy \Jeq x \lor xty \Jeq x) \lor xsy = xty\big), \\
    s \rAGGM t & \quad\text{if}\quad \forall x \forall y \colon (\rho(x) \land x \Jeq y) \to \big(xsy \Jeq x \leftrightarrow xty \Jeq x),
  \end{align*}
  using the abbreviation $\rho(x) = (\exists e \colon ee = e \land e \Jeq x)$
  to express that the $\gJ$-class of $x$ is regular.
  In view of Lemma~\ref{lem:eq} and Lemma~\ref{lem:green}, it follows
  immediately that each of the four semigroup congruences is $\FOm$-definable.
  The statement now is a consequence of Proposition~\ref{prop:congruence}.
\end{proof}

\section{Deciding Membership to First-Order Definable Classes}
\label{sec:complexity}

Our main motivation for introducing $\FOm$ over partial groupoids was providing
efficient algorithms for testing membership.
The connection between these logical formalisms and efficient membership
testing is established in this section.

\begin{proposition}
  If $\vC$ is a $\FOm$-definable class of finite partial groupoids, then the
  membership problem for $\vC$ is in $\AC^0$.
  \label{prop:ac0}
\end{proposition}
\begin{proof}
  Let $\varphi$ be a fixed $\FOm$-sentence.
  By~\cite{BarringtonIS90}, the problem of testing whether a given finite
  partial groupoid $G$ of cardinality $N$ satisfies $G \models \varphi$ is
  $\AC^0$-reducible to testing $gh = k$ for given elements $g, h, k \in G$.

  An unbounded fan-in constant-depth Boolean circuit for testing $gh = k$ works
  as follows: for each bit of the multiplication table in the input, we create
  an AND gate (with NOT gates at some of the incoming wires). The incoming
  wires to this AND gate are connected to the respective bit of the
  multiplication table as well as to all bits of the inputs $g$ and $h$; some
  of the bits corresponding to $g$ and $h$ are negated such that the gate
  evaluates to $1$ if and only if the corresponding multiplication table bit is
  $1$ and the inputs $(g, h)$ correspond to the row and column indices of the
  multiplication table entry this bit belongs to. With $\ceil{\log N}$
  additional OR gates, the outputs of these AND gates can be combined to
  compute the product $gh$. The result is then compared to $k$ using
  $2\ceil{\log N} + 1$ AND gates and $\ceil{\log N}$ OR gates. The resulting
  circuit has size $(N^2 + 4)\ceil{\log N} + 1$.
  Thus, if $\varphi$ contains only atomic formulas of the form $xy = z$, the
  membership problem is in $\AC^0$.
\end{proof}

Since $\Parity$ is not contained in $\AC^0$, we obtain the following corollary.
\begin{corollary}
  The membership problem of a $\FOm$-definable class of finite partial
  semigroups (in particular, a pseudovariety defined by a finite set of
  $\omega$-identities) cannot be hard for any complexity class containing
  $\Parity$, such as $\ACC^0$, $\TC^0$, $\LOGSPACE$ or $\NL$.
  \label{crl:omega-basis}
\end{corollary}

This yields a new technique for proving that a pseudovariety is not defined by
finitely many $\omega$-identities (by showing the stronger statement of not
even being $\FOm$-definable). We apply this technique to $\vEA$ which is
already known to be non-finitely based by~\cite{Volkov93}.
In our proof, we construct a so-called \emph{Rees matrix semigroup}.
Such a semigroup is specified by a finite group $G$ (the so-called
\emph{structure group}), two disjoint finite sets $A, B$ (called \emph{index
sets}) and a mapping $C \colon B \times A \to G \union \os{0}$ (called
\emph{sandwich matrix}). It is then defined as the set $A \times G \times B
\union \os{0}$, equipped with the multiplication
\begin{equation*}
  (a_1, g_1, b_1) (a_2, g_2, b_2) = 
    \begin{cases}
      (a_1, g_1 C(b_1, a_2) g_2, b_2) & \text{if $C(b_1, a_2) \ne 0$} \\
      0 & \text{otherwise.}
    \end{cases}
\end{equation*}
The element $0$ is a zero element.

\begin{theorem}
  The membership problem for the pseudovarieties $\vEA$, $\vA \lor \vG$ and
  $\vA * \vG$ is $\LOGSPACE$-hard (under $\AC^0$ reductions).
  \label{thm:l-hard}
\end{theorem}
\begin{proof}
  We reduce the complement of the reachability problem in undirected graphs to
  membership in $\vEA$, \ie{}we describe how to convert an undirected graph $G
  = (V, E)$ and vertices $s, t \in V$ into a Rees matrix semigroup $S$ such
  that $S \in \vEA$ if and only if $t$ is not reachable from $s$ in~$G$.
  Without loss of generality, we may assume that $s \ne t$ and $\os{s, t} \not
  \in E$.

  The structure group of $S$ is the cyclic group $C_2 = \os{1, a}$ with $a^2 =
  1$, both index sets are~$V$, and the sandwich matrix $C \colon V \times V \to
  C_2 \union \os{0}$ is given by
  \begin{equation*}
    C(v, w) =
      \begin{cases}
        1 & \text{if $v = w$ or $\os{v, w} \in E$,} \\
        a & \text{if $\os{v, w} = \os{s, t}$,} \\
        0 & \text{otherwise.}
      \end{cases}
  \end{equation*}

  The idempotent elements of $S$ are $(v, 1, v), (v, 1, w)$ with $\os{v, w} \in
  E$, as well as $(s, a, t)$, $(t, a, s)$ and the zero element $0$.
  It is clear that the reduction can be computed by $\AC^0$ circuits.
  The correctness of the construction follows almost immediately from Graham's
  Theorem; see \eg\cite[Theorem 4.13.30]{rs09qtheory}.
  A self-contained correctness proof can be found in the appendix.

  A Rees matrix semigroup belongs to $\vEA$ if and only if it belongs to $\vA
  \lor \vG$, if and only if it belongs to $\vA * \vG$; see
  \eg\cite[Theorem~4.13.31]{rs09qtheory}
  and~\cite[Corollary~4.13.32]{rs09qtheory}.
\end{proof}

\begin{corollary}
  None of the pseudovarieties $\vEA$, $\vA \lor \vG$ and $\vA * \vG$ can be
  described by a finite set of $\omega$-identities.
\end{corollary}

\begin{corollary}
  The $\oE (\cdot)$-operator does not preserve $\FOm$-definability.
\end{corollary}

The hardness result above and the fact that the membership problem for $\vEA$
is known to be decidable in polynomial time raises the question of its exact
complexity. It turns out that the problem is $\LOGSPACE$-complete.
\begin{theorem}
  The membership problem for $\vEA$ is $\LOGSPACE$-complete (under $\AC^0$
  reductions).
  \label{thm:logspace}
\end{theorem}
\begin{proof}
  In view of Theorem~\ref{thm:l-hard}, it remains to show that the membership
  problem for $\vEA$ is contained within $\LOGSPACE$.
  The proof relies on deep results on the structure of regular $\gJ$-classes
  and their idempotent-generated subsemigroups.
  We sketch the design of a deterministic log-space algorithm. Further details
  are provided in the appendix.

  Firstly, it is well-known that each regular $\gJ$-class of a finite semigroup
  $S$ is a partial subsemigroup of $S$ which is isomorphic to a Rees matrix
  semigroup; see \eg\cite[Section A.4]{rs09qtheory}.
  The construction of this Rees matrix semigroup representation for a given
  regular $\gJ$-class can be performed by a log-space transducer.
  Moreover, a finite semigroup $S$ belongs to $\vEA$ if and only if, for each
  regular $\gJ$-class~$J$ of $S$, the Rees matrix semigroup corresponding
  to~$J$ belongs to $\vEA$~\cite[Corollary 4.13.4]{rs09qtheory}.
  This observation allows us to reduce the problem to the case of Rees matrix
  semigroups:
  we iterate through all regular $\gJ$-classes of our input semigroup and can
  think of having their Rees matrix semigroup representations available.

  The \emph{incidence graph} of a Rees matrix semigroup with structure group
  $G$, index sets $A$ and $B$, and sandwich matrix $C \colon B \times A \to G
  \union \os{0}$ is defined as the edge-labelled graph with vertex set $A
  \union B$, edge set $E = \set{(b, a) \in B \times A}{C(b, a) \ne 0}$ and the
  labels as given by $C$.
  By definition, each edge label is an element from $G$.
  We consider undirected simple cycles in the incidence graph, \ie{}sequences
  $(c_1, \ldots, c_k)$ where $c_1, \ldots, c_k \in A \union B$ are pairwise
  disjoint vertices with $(c_i, c_{i+1}) \in E$ or $(c_{i+1}, c_i) \in E$ for
  all $i \in \os{1, \ldots, k-1}$, and where $c_k = c_1$.
  We extend the edge labelling to $E \union \set{(a, b)}{(b, a) \in E}$ by
  setting $C(a, b) = {(C(b, a))}^{-1}$. The label of an undirected simple cycle
  $(c_1, \dots, c_k)$ is the product $C(c_1, c_2) C(c_2, c_3) \cdots
  C(c_{k-1}, c_k)$ in $G$.
  Now, by~\cite[Corollary~4.13.24]{rs09qtheory}, a Rees matrix semigroup
  belongs to $\vEA$ if and only if each undirected simple cycle in its
  incidence graph is labelled by the identity element $1$.
  The key observation is that we do not need to consider all simple cycles: it
  suffices to choose any spanning forest $F \subseteq E$ and then compute, for
  every edge $(b, a) \in E \setminus F$, the label of the simple cycle obtained
  by starting with $(b, a)$ and then walking along spanning tree edges (and
  reverse spanning tree edges) from $a$ back to $b$.

  Again, it is easy to see that the incidence graph can be computed by a
  log-space transducer from the Rees matrix semigroup representation.
  We then (implicitly) compute the lexicographically first spanning forest $F$ of
  the incidence graph:
  checking whether an edge $(b', a')$ belongs to $F$ amounts to testing whether
  $b'$ is reachable from $a'$ in the (undirected) graph restricted to
  edges smaller than $(b', a')$\,---\,the edge-ordering is given implicitly by
  the way the graph is produced by the log-space transducer.
  For each edge $(b, a) \in E \setminus F$, we check whether the product of the
  labels of $(b, a)$ and the corresponding (reverse) spanning tree edges
  leading from $a$ to $b$ equals $1$ in~$G$.
  This is done using the standard deterministic log-space algorithm for tree
  traversal~\cite{CookM87}, multiplying the labels while walking along the
  edges.
  For a more in-depth description, we refer to the appendix.
\end{proof}

\section{Summary and Outlook}
\label{sec:summary}

We investigated the complexity of the membership problem for several well-known
pseudovarieties.
To this end, we introduced a natural formalism for defining classes of finite
partial groupoids and proved several sufficient conditions for a class to be
definable in this formalism, as well as various closure properties.
We obtained a new proof technique to show that a pseudovariety does not admit a
representation using finitely many $\omega$-identities and applied this
technique to the pseudovariety $\vEA$ by proving that its membership problem is
$\LOGSPACE$-complete.

Many interesting questions remain open.
Does $\FOm$ capture all pseudovarieties whose membership problem is in $\AC^0$?
Is $\FOm$ strictly more expressive than finite sets of $\omega$-identities?
Are there other natural pseudovarieties connected to other complexity classes
such as $\ACC^0$, $\TC^0$, $\NC^1$, $\NL$ or $\PTIME$?

As a final remark, we would like to stress that our techniques cannot be used
to decide whether a regular language given by a recognizing morphism belongs to
a certain \emph{variety of languages}.
Indeed, testing variety membership of recognized languages, is much harder.

\begin{proposition}
  Let $\vV \subsetneq \vS$ be a pseudovariety of finite semigroups, let $h
  \colon A^+ \to S$ be a morphism to a finite semigroup $S$ and let $P
  \subseteq S$.
  Then the problem of deciding whether $h^{-1}(P)$ is recognized by a semigroup
  in $\vV$ is $\NL$-hard (under $\AC^0$ reductions).
  \label{prop:languages}
\end{proposition}
\begin{proof}
  We prove $\NL$-hardness of the problem by a reduction from the complement of
  $s$-$t$-connectivity in directed graphs.
  Let $G = (V, E)$ be a directed graph and let $s, t \in V$ with $s \ne t$.

  Let $L \subseteq B^+$ be a regular language which is not recognized by any
  semigroup from $\vV$ and let $g \colon B^+ \to T$ be a morphism to a finite
  semigroup $T \not\in \vV$ and $Q \subseteq T$ such that $g^{-1}(Q) = L$.
  Such a language and such a morphism exist by Eilenberg's Theorem on
  \emph{$+$-varieties of languages}~\cite[Theorem 3.4s]{eil76} and the assumption
  $\vV \ne \vS$.
  We extend $g$ to a monoid morphism $g \colon A^* \to T^1$.
  Below, the symbol $\circ$ will be used to denote the multiplication in $T^1$.

  We define a semigroup $S = V \times T^1 \times V \union \os{0}$ by the binary
  operation
  \begin{equation*}
    (v, m, w) \cdot (x, n, y) =
      \begin{cases}
        (v,  m \circ n, y) & \text{if $w = x$ and $x \ne t \lor y = t$,} \\
        0 & \text{otherwise}.
      \end{cases}
  \end{equation*}
  The element $0$ is a zero element.
  Let $A = \set{(v, \varepsilon, w)}{(v, w) \in E} \union \set{(t, b, t)}{b \in
  B}$ and let $P = \set{(s, m, t)}{m \in Q}$.
  Let $h \colon A^+ \to S$ be the mapping defined by $h(v, b, w) = (v, g(b), w)$.

  If $t$ is not reachable from $s$ in $G$, we have $h^{-1}(P) = \emptyset$. The
  empty language is recognized by every finite semigroup. In particular,
  $h^{-1}(P)$ is recognized by a semigroup from $\vV$.

  Now, suppose that there exists a directed path $(x_1, \dots, x_k)$ with $x_1 =
  s$ and $x_k = t$. We may assume that $t \not\in \os{x_1, \dots, x_{k-1}}$.
  Assume, for the sake of contradiction, that the language $h^{-1}(P)$ is
  recognized by a finite semigroup $T' \in \vV$, \ie{}there exist a morphism $f
  \colon A^+ \to T'$ and a set $R \subseteq T'$ with $f^{-1}(R) = h^{-1}(P)$.
  We let $g' \colon B^+ \to T'$ be the morphism defined by $g'(b) = f(t, b, t)$
  for all $b \in B$ and we let $Q' = \set{r \in T'}{f\big((x_1, \varepsilon, x_2)
  \cdots (x_{k-1}, \varepsilon, x_k)\big)r \in R}$.
  We obtain
  \begin{align*}
    g(b_1 \cdots b_k) \in Q & \quad\Leftrightarrow\quad h\big((x_1, \varepsilon, x_2) \cdots (x_{k-1}, \varepsilon, x_k) \ms (t, b_1, t) \cdots (t, b_k, t)\big) \in P \\
                            & \quad\Leftrightarrow\quad f\big((x_1, \varepsilon, x_2) \cdots (x_{k-1}, \varepsilon, x_k) \ms (t, b_1, t) \cdots (t, b_k, t)\big) \in R \\
                            & \quad\Leftrightarrow\quad g'(b_1 \cdots b_k) \in Q'
  \end{align*}
  for all sequences of elements $b_1, \dots, b_k \in B$.
  Thus, the preimage of $Q'$ under $g' \colon B^+ \to T'$ is~$L$, contradicting
  the choice of $L$.
\end{proof}

\subparagraph{Acknowledgements.} I would like to thank Charles Paperman for
suggesting the statement of Lemma~\ref{lem:omega-green} which led to the proof
that all pseudovarieties definable by finite sets of $\omega$-identities are
$\FOm$-definable.

\newcommand{\Ju}{Ju}\newcommand{\Ph}{Ph}\newcommand{\Th}{Th}\newcommand{\Ch}{Ch}\newcommand{\Yu}{Yu}\newcommand{\Zh}{Zh}\newcommand{\St}{St}\newcommand{\curlybraces}[1]{\{#1\}}

\newpage
\appendix

\section{Additions to Section~\ref{sec:complexity}}

\subsection{Correctness of the Construction in Theorem~\ref{thm:l-hard}}

We show that the Rees matrix semigroup $S$ with structure group $C_2 = \os{1,
a}$, index sets~$V$, and sandwich matrix $C \colon V \times V \to C_2 \union
\os{0}$ given by
\begin{equation*}
  C(v, w) =
    \begin{cases}
      1 & \text{if $v = w$ or $\os{v, w} \in E$,} \\
      a & \text{if $\os{v, w} = \os{s, t}$,} \\
      0 & \text{otherwise,}
    \end{cases}
\end{equation*}
is contained in $\vEA$ if and only if in the underlying graph $G = (V, E)$, the
vertex $t$ is not reachable from $s$. Remember that by assumption, we have
$s \ne t$ and $\os{s, t} \not\in E$.

Let $(v, g, w)$ with $v, w \in V$ and with $g \in C_2$. When computing $(v, g,
w)(v, g, w)$, there are three possibilities. If $v = w$ or $\os{v, w} \in E$,
we obtain $(v, g, w)(v, g, w) = (v, g^2, w)$ which equals $(v, g, w)$ if and
only if $g = 1$. If $\os{v, w} = \os{s, t}$, then $(v, g, w)(v, g, w) = (v,
gag, w)$ which equals $(v, g, w)$ if and only if $g = a$. If neither $v = w$
nor $\os{v, w} \in E$ nor $\os{v, w} = \os{s,t}$, we have $(v, g, w)(v, g, w) =
0$.
Thus, the only idempotent elements of $S$ are the elements $(v, 1, v)$ with $v
\in V$, the elements $(v, 1, w)$ with $\os{v, w} \in E$, as well as $(s, a,
t)$, $(t, a, s)$ and $0$.

We will now show that the following properties are equivalent:
\begin{enumerate}
  \item\label{enum:aaa} $S \not\in \vEA$, \ie$\langle E(S) \rangle \not\in \vA$.
  \item\label{enum:bbb} There exist $x_1, \ldots, x_k \in V$ such that $(x_1, 1, x_1) \cdots (x_k, 1, x_k) \in \os{(s, 1, t), (t, 1, s)}$.
  \item\label{enum:ccc} There exists a path $(x_1, \dots, x_k)$ in $G$ with $\os{x_1, x_k} = \os{s, t}$.
\end{enumerate}

\paragraph{(\ref{enum:aaa}) implies (\ref{enum:bbb}).}
First, note that $(v, g, w) = (v, 1, v) (w, 1, w)$ for all $(v, g, w) \in
E(S)$. Thus, we have $\langle E(S) \rangle = \langle \set{(v, 1, v)}{v \in V}
\union \os{0} \rangle$.
Now, if $\langle E(S) \rangle \not\in \vA$, there exist two elements%
\begin{equation*}
  (v, g, w) = (v_1, 1, v_1) \cdots (v_k, 1, v_k) \quad\text{and}\quad
  (x, h, y) = (x_1, 1, x_1) \cdots (x_\ell, 1, x_\ell)
\end{equation*}
such that $(v, g, w) \Heq (x, h, y)$ and $(v, g, w) \ne (x, h, y)$, \ie{}$v =
x$ and $w = y$ and $g \ne h$.
We obtain
\begin{equation*}
  (v_1, 1, v_1) \cdots (v_1, 1, v_k) (x_\ell, 1, x_\ell) \cdots (x_1, 1, x_1) = (v, g, w) (y, h, x) = (v, a, v)
\end{equation*}
where the first equality holds by symmetry of $C$ (and from the fact that $C_2$
is Abelian) and the second equality follows from $v = x$, $w = y$ and from $g
\ne h$, \ie$gh = a$.

To simplify notation, we let $v_{k+1+i} = x_{\ell-i}$ for $i \in \os{0, \dots,
\ell-1}$.
Since the product $(v_1, 1, v_1) \cdots (v_{k+\ell}, 1, v_{k+\ell})$ equals
$(v, a, v)$, there exists some $i \in \os{1, \dots, k+\ell-1}$ with $\os{v_i,
v_{i+1}} = \os{s, t}$. Therefore,
\begin{equation*}
  (v_{i+1}, 1, v_{i+1}) \cdots (v_{k+\ell}, 1, v_{k+\ell}) (v_1, 1, v_1) \cdots (v_i, 1, v_i) \in \os{(s, 1, t), (t, 1, s)},
\end{equation*}
as desired.

\paragraph{(\ref{enum:bbb}) implies (\ref{enum:ccc}).}
Suppose that $(x_1, 1, x_1) \cdots (x_k, 1, x_k) \in \os{(s, 1, t), (t, 1, s)}$.
We apply induction on the cardinality of the set $I$ of indices $i \in \os{1,
\dots, k-1}$ with $\os{x_i, x_{i+1}} = \os{s, t}$.
Note that $\abs{I}$ is even, otherwise $(x_1, 1, x_1) \cdots (x_k, 1, x_k) \in
\os{(s, a, t), (t, a, s)}$ by the definition of $S$.
Without loss of generality, we may also assume that $x_i \ne x_{i+1}$ for $1
\le i < k$.

If $\abs{I} = 0$, then $(x_1, \dots, x_k)$ is a path connecting $s$ and $t$ in
$G$.

If $\abs{I} \ge 2$, consider two indices $i, j \in I$ such that $i < j$ and
$\os{i+1, \dots, j-1} \intersect I = \emptyset$.
If $x_i = x_j$, then $(x_{i+1}, \dots, x_j)$ is a path connecting $s$ and $t$
in $G$.
If $x_i \ne x_j$, it is easy to verify that the product $(x_i, 1, x_i) \cdots
(x_{j+1}, 1, x_{j+1})$ equals $(x_i, 1, x_i)$. Thus,
\begin{equation*}
  (x_1, 1, x_1) \cdots (x_i, 1, x_i) (x_{j+2}, 1, x_{j+2}) \cdots (x_k, 1, x_k) \in \os{(s, 1, t), (t, 1, s)}
\end{equation*}
and we obtain the statement by induction on $\abs{I}$.

\paragraph{(\ref{enum:ccc}) implies (\ref{enum:aaa}).}
Suppose that $(x_1, \dots, x_k)$ is a path in $G$ with $x_1 = s$ and $x_k = t$.
Then, $(x_1, 1, x_1) \cdots (x_k, 1, x_k) = (s, 1, t)$ in $S$. Thus, $(s, 1, t)
\in \langle E(S) \rangle$. Since $(s, a, t) \in E(S)$ and $(s, 1, t) \Heq (s,
a, t)$, this yields $\langle E(S) \rangle \not\in \vA$.
The case $x_1 = t$ and $x_k = s$ is symmetric.
\qed%

\subsection{Additions to the Proof of Theorem~\ref{thm:logspace}}

Some details were skipped in the proof of Theorem~\ref{thm:logspace} given in
Section~\ref{sec:complexity}. We try to fill in the missing gaps here.
The input is a finite semigroup $S$ of cardinality $N$.

\paragraph{Computing the Rees Matrix Representation.}

The first step in the proof is a reduction to the case of Rees matrix
semigroups.
This is based on the observation that, for every idempotent element $e \in
E(S)$, one can construct a Rees matrix semigroup which is isomorphic to the
finite partial semigroup corresponding to the $\gJ$-class containing $e$.
Of course, the result of the construction cannot be stored on the work tape but
it suffices to be able to enumerate its elements and perform single
multiplications within logarithmic space.
Being able to do this, we can iterate over all elements of the input semigroup
$S$ (using a single counter with $\ceil{\log N}$ bits), and for each idempotent
element, we run the algorithm Rees matrix semigroup algorithm on the implicit
Rees matrix semigroup representation.

To construct the Rees matrix semigroup representation for the $\gJ$-class of a
given idempotent element $e \in E(S)$, we make some preliminary considerations.
Firstly, there is a natural order on the elements of~$S$ given by their
encoding in the input, \ie{}we say that an element $s \in S$ is
\emph{(strictly) smaller} than $t \in S$ if their corresponding encodings as
binary numbers satisfy $s < t$.
Note that one can easily check within logarithmic space whether two elements
$s, t \in S$ satisfy any of Green's relations ($s \Req t$, $s \Leq t$, $s \Heq
t$ or $s \Jeq t$).

The index set $A$ of our Rees matrix semigroup for the $\gJ$-class of $e$ then
consists of all elements $a \in S$ such that $a \Leq e$ and such that for all
elements $a' < a$, we have $a' \not\Req a$, \ie{}we pick minimal elements from
each $\gH$-class contained in the $\gL$-class of~$e$. Roughly speaking, the
elements of $A$ correspond to the $\gR$-classes of $S$.
Analogously, the index set $B$ consists of all minimal elements of
$\gH$-classes contained in the $\gR$-class of~$e$.
The $\gH$-class of $e$ itself is the structure group $G$.
The sandwich matrix $C \colon B \times A \to G \union \os{0}$ is given by $C(b,
a) = ba$ if $ba \in G$ (\ie$ba \Heq e$) and $C(b, a) = 0$ otherwise.

Testing whether a tuple $(a, g, b) \in S^3$ belongs to the Rees matrix
semigroup $A \times G \times B \union \os{0}$ and performing multiplications
$(a, g, b)(c, h, d)$ now both reduces to multiplications in $S$, testing
Green's relations in $S$ and comparing integers.
A routine calculation shows that the Rees matrix semigroup is indeed isomorphic
to the partial semigroup obtained by restricting $S$ to the $\gJ$-class of $e$;
the zero element of the Rees matrix semigroup corresponds to the ``undefined
value'' in this partial semigroup. Details are available, e.g.,
in~\cite[Section~A.4]{rs09qtheory}.

\paragraph{Computing the Incidence Graph.}

An implicit representation of the incidence graph can be computed in a fashion
similar to the construction of the Rees matrix semigroup above.
Remember that the vertices of this graph are $A \union B$. As observed in the
construction of the Rees matrix semigroup, testing whether an element $s \in S$
belongs to $A$ or to $B$ is attainable in logarithmic space.
Checking whether $(b, a) \in B \times A$ is an edge in the incidence graph
reduces to computing the product $ba$ in $S$ and testing whether the result is
$\gH$-equivalent to the idempotent $e$ of the currently considered $\gJ$-class.
The label $C(b, a)$ of an edge is obtained by computing the product $ba$ in $S$
and its inverse $C(a, b) = {(C(b, a))}^{-1}$ is obtained by iterating over all
elements of $S$ until an element $s$ with $s \Heq ba$ and $bas = e = sba$ is
found.

\paragraph{Checking Cycles.}

A lexicographically first spanning forest of the incidence graph can be
computed as already mentioned in the proof of the theorem: to check whether an
edge $(b, a)$ belongs to the spanning forest, we run a deterministic log-space
algorithm for undirected graph reachability on the (undirected) incidence
graph, pretending that edges $(b', a')$ with $b' > b$ or with $b' = b \land a'
> a$ do not exist.
The existence of a deterministic log-space algorithm for undirected graph
reachability was proven in 2008 by Omer Reingold~\cite{Reingold08}.

To enumerate all \emph{bridges} (\ie{}edges of the incidence graph which are
not part of this spanning forest), we iterate over all tuples $(b, a) \in S^2$,
check for each tuple whether $(b, a) \in B \times A$ and then ask whether or
not $(b, a)$ belongs to the spanning forest as described above.
The process of walking along the simple undirected cycle \emph{induced by a
bridge $(b, a)$}, that is, the unique simple cycle consisting of the edge $(b,
a)$ and spanning tree edges (and reverse spanning tree edges) from $a$ to $b$
works as follows:
when starting at the vertex $a$, we look at all incident edges and take the
lexicographically smallest among those which are part of the spanning forest.
We walk along this edge.
Whenever we arrive at a vertex~$c$, we again look at all incident edges and
choose the lexicographically smallest edge which belongs to the spanning forest
and is larger than the edge we just used to enter $c$. If no such edge exists,
we choose the lexicographically smallest edge incident to $c$ instead.
Note that if we arrive at a leaf, we immediately go back to where we just came
from.
The process can be thought of as tracing an Euler tour around the spanning
tree.
We eventually arrive at $b$ because it belongs to the same connected component
as $a$.
Whenever walking along an edge $(b', a')$, we multiply the currently stored
label of the path by the label $C(b', a')$ of the edge. When walking along such
an edge in the opposite direction, we multiply the currently stored label by
$C(a', b')$. Since $C(b', a')$ and $C(a', b')$ are mutually inverse, we
undo ``wrong'' multiplications in backtracking steps.

The fact that it suffices to compute the labels of simple undirected cycles
induced by bridges is based on the following idea: firstly, an undirected cycle
$(c_1, \dots, c_k)$ is labelled by $1$ if and only if its reverse $(c_k, \dots,
c_1)$ is labelled by $1$. Now suppose that every simple undirected cycle
induced by a bridge is labeled by $1$ (which is encoded by the element $e \in
S$).
Then, the label of any bridge $(b, a)$ is the same as the label of the shortest
(undirected) path in the spanning forest leading from $b$ to $a$. Analogously,
the label of any pair $(a, b)$ for a bridge $(b, a)$ is the same as the label
of the shortest (undirected) path in the spanning forest leading from $a$ to
$b$.
Therefore, given any simple undirected cycle, we can repeatedly replace bridges
(or reverse bridges) in this cycle by a sequence of (partly reverse) spanning
tree edges without changing its label, until only a single bridge remains.
After each substitution, subpaths of the form $(c_1, \dots, c_{k-1}, c_k,
c_{k-1}, \dots, c_1)$, so called \emph{backtracks}, can be removed.
For further details, we recommend reading Section~4.13.1 of~\cite{rs09qtheory}.

\end{document}